\documentclass[a4paper,12pt]{article}
\usepackage[russian, english]{babel}
\usepackage{amsmath}
\usepackage{ulem}
\usepackage{calc}
\usepackage{vmargin}
\usepackage{amstext}
\usepackage{graphicx}
\usepackage{amsfonts}
\usepackage{amssymb}
\usepackage{verbatim}

\usepackage{amsthm}
\newtheorem{theorem}{Theorem}[section]
\newtheorem{proposition}[theorem]{Proposition}

\theoremstyle{definition}

\begin{document}
		
		\title{\bf Exact Harmonic Dimensional Reduction and Conformal Lifting for Multicomponent $(3+1)$ Nonlinear Schrödinger Systems}

	%		\bf Exact Harmonic Dimensional Reduction of Multicomponent $(3+1)\mathrm{D}$ Nonlinear Schrödinger Systems and Gross--Pitaevskii Vortex Filaments}
		
		\author{O.~V. Kaptsov
			 \\ Federal Research Center for Information and Computational Technologies,
			 \\ Novosibirsk, Russia
			\\ Email: kaptsov@ict.nsc.ru}
		
		\date{}
		\maketitle
		
		\numberwithin{equation}{section}
		
		\begin{abstract}
			A harmonic   dimensional reduction framework is developed for
			$(3+1)\mathrm{D}$ systems of coupled nonlinear Schr\"odinger-type
			equations with stationary transverse trapping potentials.
			The central result is a lifting lemma: if the transverse phase
			functions are harmonic and the trapping potential exactly cancels
			the squared phase gradient, the full $(3+1)\mathrm{D}$ system
			reduces identically to a closed $(1+1)\mathrm{D}$ integrable
			hierarchy, and every solution of the reduced system lifts to an
			exact solution of the original multidimensional model.
		The framework is applied to four systems.
		For the scalar Gross--Pitaevskii equation, Kuznetsov--Ma breathers
		are embedded in $(3+1)\mathrm{D}$ geometries carrying vortex
		lattices with finite, non-singular density at the cores.
		For the two-component Manakov system, the phase-inversion ansatz
		yields exact vector solutions with vanishing mass current and
		non-trivial transverse spin current modulated by the longitudinal
		breather.
		For the three-component spinor $F=1$ Bose--Einstein condensate,
		a symmetric Kuznetsov--Ma breather and a spin-exchange rogue wave
		are constructed, the latter exhibiting transient density
		amplification by a factor of nine in the $m_F=0$ channel.
		For the Maxwell--Bloch system, self-induced transparency solitons,
		two-soliton elastic collisions, and Kuznetsov--Ma breathers are
		lifted to full $(3+1)\mathrm{D}$ geometry, with population
		inversion remaining transversely uniform despite arbitrary phase
		winding in the cross-section. 
		\end{abstract}
			
	 {\bf Keywords:} Gross--Pitaevskii equation,
	Nonlinear Schr\"odinger equation, Exact analytical solutions,
	Spinor Bose--Einstein condensate, Maxwell--Bloch system,
	Kuznetsov--Ma breather, Vortex lattice.

	\section{Introduction}
	\label{sec:introduction}
	
	The analytical description of localized wave structures in
	multi-dimensional configurations remains one of the most prominent
	challenges in mathematical physics. While one-dimensional systems
	governed by the canonical nonlinear Schr\"odinger (NLS) equation
	are celebrated for their strict integrability via the Inverse
	Scattering Transform~\cite{Zakharov1972, Ablowitz1981, Zakharov2009},
	their higher-dimensional counterparts exhibit a fundamentally
	different physical behavior. In $(2+1)\mathrm{D}$ and
	$(3+1)\mathrm{D}$ spaces, attractive cubic nonlinearities
	inevitably trigger critical wave collapse and self-focusing
	instabilities~\cite{Faddeev1987, Kivshar2003}. This phenomenon
	leads to finite-time singularities, which preclude the structural
	stability of multidimensional solitary waves and seriously
	complicate the accurate analytical mapping of quantum states.
	
	To stabilize localized structures against collapse, external
	potentials are routinely introduced. In the context of ultracold
	atomic physics, magneto-optical trapping fields and time-averaged
	orbiting potentials are widely utilized to confine quantum liquids
	governed by the Gross--Pitaevskii (GP)
	equation~\cite{Pitaevskii2003, Dalfovo1999}. Similarly, periodic
	photonic crystals and graded-index waveguides are actively employed
	in nonlinear optics to counteract transverse
	diffraction~\cite{Kartashov2011, Malomed2005}. In multi-component or spinor
	quantum gases, the field dynamics are further complicated by
	cross-phase modulation channels and coherent spin-exchange
	mechanisms between internal atomic
	states~\cite{Manakov1974, Tsuchida1998, Ho1998, OhmiMachida1998, Kawaguchi2012}.
	A similarly demanding analytical problem arises when nonlinear
	wave envelopes couple with resonant two-level media possessing
	coherent phase memory, described via Maxwell--Bloch
	equations~\cite{McCall1967, Ablowitz1981}.
	
	The vast majority of existing theoretical frameworks rely on
	variational approximations, numerical simulations, or perturbative
	tools. Although these methods provide valuable qualitative
	insights, they frequently obscure the exact geometric and algebraic
	symmetries hidden within the governing partial differential
	equations. Developing rigorous approaches capable of yielding
	exact analytical solutions for multi-dimensional nonlinear models
	therefore remains a critical objective. In this regard, Lie group
	analysis, functional separation of variables, and systematic
	algebraic reduction methods occupy a distinct position, as they
	allow for a structured exact reduction of spatial dimensionality.
	
	In this work, we present a generalized exact harmonic (conformal) dimensional
	reduction method applicable to a broad class of multi-component
	$(3+1)\mathrm{D}$ systems of Schr\"odinger type. The central
	result is a lifting lemma: if the transverse phase functions are
	harmonic and the trapping potential exactly cancels the squared
	phase gradient, the full $(3+1)\mathrm{D}$ system reduces
	identically to a closed $(1+1)\mathrm{D}$ integrable hierarchy.
	Every solution of the reduced system then lifts to an exact
	solution of the original multidimensional model via the conformal
	ansatz $\psi_n = A_n(t,x)\,e^{iv_n(y,z)}$. This construction
	inherently precludes wave collapse: the density $|\psi_n|^2 =
	|A_n(t,x)|^2$ is bounded by the one-dimensional amplitude alone,
	independently of the transverse geometry.
	
	The framework is applied to four physically distinct systems.
	For the scalar Gross--Pitaevskii equation, Kuznetsov--Ma
	breathers~\cite{Kuznetsov1977, Ma1979} are embedded in
	$(3+1)\mathrm{D}$ geometries carrying Weierstrass and Jacobi
	vortex lattices. The resulting vortex filaments possess a
	fundamental distinction from conventional Bose--Einstein
	condensate vortices: the density remains strictly finite at the
	cores, since the singular kinetic energy of the phase gradient
	(scaling as $+1/r^2$) is exactly cancelled by a matching singular
	attractive trapping potential (scaling as $-1/r^2$).
	For the two-component Manakov system, a phase-inversion ansatz
	yields exact vector solutions with identically vanishing mass
	current and non-trivial transverse spin current modulated by the
	longitudinal breather.
	For the three-component spinor $F=1$ Bose--Einstein condensate
	governed by the Ieda--Miyakawa--Wadati
	system~\cite{Ieda2004PRL, Ieda2004JPSJ, Li2018}, two exact
	solution families are constructed: a symmetric Kuznetsov--Ma
	breather on the fully magnetized background, and a spin-exchange
	rogue wave on the unmagnetized background exhibiting transient
	density amplification by a factor of nine in the $m_F = 0$
	channel.
	For the Maxwell--Bloch system governing resonant pulse propagation
	in doped dielectric waveguides and ferromagnetic
	films~\cite{McCall1967, Ablowitz1981}, self-induced transparency
	solitons, two-soliton elastic collisions, and Kuznetsov--Ma
	breathers are lifted to full $(3+1)\mathrm{D}$ geometry. In all
	cases the population inversion remains transversely uniform
	despite arbitrary phase winding in the beam cross-section, a
	direct consequence of the factorized ansatz.

%%% NEW NEW

\section{Harmonic reduction and lifting map.}
\label{sec:reduction}

Consider a $(3+1)\mathrm{D}$ system of coupled nonlinear evolution equations of Schrödinger type, operating under the action of independent, component-wise transverse trapping potentials $V_n(y,z)$:
\begin{equation}
	\left( i \frac{\partial}{\partial t} + \frac{1}{2} \nabla^2 
	- V_n(y,z) \right) \psi_n 
	+ \mathbf{F}_n(\vec{\psi}, \vec{\psi}^*) = 0,
	\qquad n = 1, \dots, N,
	\label{eq:full_system}
\end{equation}
where $\nabla^2 = \partial^2_x + \Delta_{\perp}$ is the full three-dimensional Laplacian, $\Delta_{\perp} = \partial^2_y + \partial^2_z$ is the transverse Laplacian, and $\mathbf{F}_n$ represents an arbitrary nonlinear interaction operator.

We seek exact solutions confined within a specific spatial domain $\mathbb{R}^3$. We propose a separable, phase-modulated ansatz of the form
\begin{equation}
	\psi_n(t,x,y,z) = A_n(t,x)\, e^{i v_n(y,z)},
	\qquad n = 1, \dots, N,
	\label{eq:ansatz}
\end{equation}
where the complex amplitudes $A_n$ depend exclusively on the longitudinal and temporal coordinates $(t,x)$, while the real transverse phase functions $v_n(y,z)$ are  harmonic
\begin{equation}
	\Delta_{\perp} v_n 
	\equiv \frac{\partial^2 v_n}{\partial y^2} 
	+ \frac{\partial^2 v_n}{\partial z^2} = 0,
	\label{eq:harmonic}
\end{equation}

To ensure that the nonlinear interaction does not reintroduce the transverse coordinates into the longitudinal dynamics, we assume that the nonlinear operator satisfies the componentwise gauge covariance property
\begin{equation}
	\mathbf{F}_n\!\left(\vec{A} e^{i\vec{v}},\,
	\vec{A}^* e^{-i\vec{v}}\right)
	= \mathbf{F}_n^{(1\mathrm{D})}(\vec{A}, \vec{A}^*)\,
	e^{i v_n(y,z)}.
	\label{eq:homogeneity}
\end{equation}
This property means that each nonlinear term transforms covariantly under independent local phase shifts of the field components and therefore preserves the separable structure of the ansatz.
 In particular, for the standard cubic cross-phase modulation interaction $\mathbf{F}_n = \sum_m g_{nm}|\psi_m|^2 \psi_n$, substitution yields
\begin{equation}
	\mathbf{F}_n\!\left(\vec{A} e^{i\vec{v}},\vec{A}^* e^{-i\vec{v}}\right)
	= \sum_m g_{nm}|A_m e^{i v_m}|^2 A_n e^{i v_n}
	= \sum_m g_{nm}|A_m|^2 A_n \cdot e^{i v_n}.
	\label{eq:cubic_check}
\end{equation}
Thus, $\mathbf{F}_n^{(1\mathrm{D})} = \sum_m g_{nm}|A_m|^2 A_n$, which coincides with the required structural condition~\eqref{eq:homogeneity}.

Finally, we assume that each  external potential balances the squared gradient of its corresponding phase function
% up to a constant chemical-potential shift $\rho_n \in\mathbb{R}$
\begin{equation}
	V_n(y,z) = -\frac{1}{2}|\nabla_{\perp} v_n|^2 - \rho_n.
	\label{eq:potential}
\end{equation}
The constants $\rho_n$ define real spectral shift parameters.

\noindent
\textbf{Lemma (Harmonic reduction).}
\label{lem:reduction}
Under the conditions~\eqref{eq:harmonic}--\eqref{eq:potential}, the full multidimensional $(3+1)\mathrm{D}$ system~\eqref{eq:full_system} reduces identically within the domain $\mathbb{R}^3$ to the closed $(1+1)\mathrm{D}$ system:
\begin{equation}
	i \frac{\partial A_n}{\partial t}
	+ \frac{1}{2} \frac{\partial^2 A_n}{\partial x^2}
	+ \mathbf{F}_n^{(1\mathrm{D})}(\vec{A}, \vec{A}^*)
	+ \rho_n A_n = 0.
	\label{eq:reduced_system}
\end{equation}

\medskip
\noindent\textbf{Proof.}
Since $A_n = A_n(t,x)$, all transverse derivatives vanish:
$\nabla_{\perp} A_n = 0$. Applying the product rule for $\Delta_{\perp}$ gives:
% Evaluating the action of the full spatial Laplacian operator $\nabla^2 = \partial^2_x + \Delta_{\perp}$ on the ansatz~\eqref{eq:ansatz} via the standard differentiation product rule yields:
%Since $A_n = A_n(t,x)$, all transverse derivatives of the amplitude 
%vanish: $\nabla_{\perp} A_n = 0$. 
%Applying the product rule for $\Delta_{\perp}$ gives:
\begin{equation}
	\Delta_{\perp}(A_n e^{iv_n}) 
	= (\Delta_{\perp} A_n) e^{iv_n} + 2 (\nabla_{\perp} A_n) \cdot (\nabla_{\perp} e^{iv_n}) + A_n \Delta_{\perp}(e^{iv_n}).
	\label{eq:leibniz}
\end{equation}
%Since the first two terms on the right-hand side vanish (because $\nabla_{\perp} A_n = 0$ and $\Delta_{\perp} A_n = 0$), this simplifies to $\Delta_{\perp}(A_n e^{iv_n}) = A_n \Delta_{\perp}(e^{iv_n})$, and 
Consequently the total Laplacian acting on the ansatz is
\begin{equation}
	\nabla^2(A_n e^{iv_n}) 
	= (\partial_x^2 A_n)\,e^{iv_n} + A_n\,\Delta_{\perp}(e^{iv_n}).
	\label{eq:laplacian_expansion_step1}
\end{equation}
Expanding the transverse Laplacian of the exponential phase factor leads to:
\begin{equation}
	\Delta_{\perp} \left( e^{i v_n} \right) 
	= \left( i \Delta_{\perp} v_n - |\nabla_{\perp} v_n|^2 \right) e^{i v_n}.
	\label{eq:laplacian_expansion_step2}
\end{equation}
Substituting the strict harmonicity condition~\eqref{eq:harmonic} into Eq.~\eqref{eq:laplacian_expansion_step2}, the imaginary part vanishes identically for all $(y,z)\in\mathbb{R}^2$. Hence, the total Laplacian simplifies to a purely real expression:
\begin{equation}
	\nabla^2 \psi_n = \left( \frac{\partial^2 A_n}{\partial x^2} - |\nabla_{\perp} v_n|^2 A_n \right) e^{i v_n}.
	\label{eq:laplacian_expansion}
\end{equation}
Substituting the reduced Laplacian~\eqref{eq:laplacian_expansion} and the gauge homogeneity property~\eqref{eq:homogeneity} back into the initial equations~\eqref{eq:full_system}, we factor out the non-vanishing exponential terms $e^{i v_n} \neq 0$:
\begin{equation}
	i \frac{\partial A_n}{\partial t}
	+ \frac{1}{2}\frac{\partial^2 A_n}{\partial x^2}
	- \frac{1}{2}|\nabla_{\perp} v_n|^2 A_n
	- V_n(y,z)\,A_n
	+ \mathbf{F}_n^{(1\mathrm{D})}(\vec{A}, \vec{A}^*)
	= 0.
	\label{eq:factored_form}
\end{equation}
The transverse phase gradients generate an effective kinetic-energy contribution
\[
\frac12|\nabla_\perp v_n|^2,
\]
which acts as a geometric potential in the reduced equations. We therefore choose the potential to compensate this contribution exactly:
%Finally, substituting the explicit profile of the trapping potential~\eqref{eq:potential} into Eq.~\eqref{eq:factored_form} leads to the exact cancellation of the gradient terms:
\begin{equation}
	-\frac{1}{2}|\nabla_{\perp} v_n|^2 A_n - \left( -\frac{1}{2}|\nabla_{\perp} v_n|^2 - \rho_n \right) A_n = \rho_n A_n.
	\label{eq:cancellation}
\end{equation}
This reduces~\eqref{eq:factored_form} to~\eqref{eq:reduced_system}.\qed

The proof remains valid for arbitrary harmonic phase functions,
including single-valued branches of multivalued harmonic functions,
provided they are restricted to simply connected domains excluding
their singular points.

\noindent
{\bf Corollary (Exact lifting operator).}

The mapping
\[
\mathcal{L}:
\{A_n\}
\longmapsto
\{A_n e^{iv_n}\}
\]
defines an exact lifting operator from the solution space of the reduced
$(1+1)$-dimensional system \eqref{eq:reduced_system} into the solution
space of the original multidimensional system \eqref{eq:full_system}.
Consequently, every exact solution of the reduced model generates an
exact solution of the full $(3+1)$-dimensional equations.

\section{Scalar Gross--Pitaevskii Equation: Exact Solutions}
\label{sec:scalar_GP}

As a first application of the harmonic lifting theorem, we consider the
scalar Gross--Pitaevskii equation
\begin{equation}
	i \frac{\partial \psi}{\partial t} 
	+ \frac{1}{2}\nabla^2\psi
	- V(y,z)\,\psi 
	+ |\psi|^2 \psi = 0,
	\label{eq:scalar_GP}
\end{equation}
where the nonlinear coefficient is set to $\gamma = 1$ (self-focusing medium or attractive BEC).
The reduction established in Section~2 immediately yields an exact
correspondence between the multidimensional GP equation and the standard
one-dimensional focusing nonlinear Schrödinger equation.

Applying Lemma with the scalar ansatz $\psi = A(t,x)\,e^{iv(y,z)}$, where $\Delta_\perp v = 0$ and $V(y,z) = -\tfrac{1}{2}|\nabla_\perp v|^2 - \rho$, the cubic nonlinearity reduces as
\begin{equation}
	|\psi|^2\psi = |A|^2 A \cdot e^{iv},
	\label{eq:nonlinearity_reduction}
\end{equation}
confirming the $U(1)$-equivariance condition~\eqref{eq:homogeneity} for $\mathbf{F}^{(1\mathrm{D})} = |A|^2 A$.
The $(3+1)\mathrm{D}$ GP equation thereby reduces exactly to the $(1+1)\mathrm{D}$ NLS equation with a constant energy shift:
\begin{equation}
	i \frac{\partial A}{\partial t} 
	+ \frac{1}{2}\frac{\partial^2 A}{\partial x^2} 
	+ |A|^2 A + \rho A = 0.
	\label{eq:reduced_NLSE}
\end{equation}

\begin{proposition}
	\label{prop:gauge}
	The gauge transformation
	\[
	A=e^{i\rho t}\Psi
	\]
	removes the constant energy shift from
	\eqref{eq:reduced_NLSE}
	and transforms it into the standard focusing NLS equation
	$$   	i \frac{\partial \Psi}{\partial t} 
	+ \frac{1}{2}\frac{\partial^2  \Psi}{\partial x^2} + |\Psi|^2  \Psi  = 0.$$
\end{proposition}

The proof is carried out by direct substitution.

Consequently every exact solution of the focusing NLS immediately
produces an exact solution of the multidimensional GP equation.

As an illustration, we employ the classical
Kuznetsov--Ma breather
\cite{Kuznetsov1977,Ma1979,Akhmediev1987},
which is an exact solution of the focusing NLS equation.
Applying the gauge transformation of Proposition~3.1 to the  Kuznetsov--Ma breather yields the exact solution of equation~\eqref{eq:reduced_NLSE}:

\begin{equation}
	A(t,x)
	=
	\sqrt{2}
	\left[
	1-
	\frac{a\bigl(a\cos(\Omega t)
		+i\sqrt{a^{2}+8}\,\sin(\Omega t)\bigr)}
	{\sqrt{2}\sqrt{a^{2}+8}\,\cosh(ax)
		-4\cos(\Omega t)}
	\right]
	e^{i(\rho+2)t},
	\label{eq:KM_general}
\end{equation}
where
\begin{equation}
	\Omega
	=
	\frac{a}{2}\sqrt{a^{2}+8},
	\qquad
	a>0.
	\label{eq:dispersion}
\end{equation}

The solution is periodic in time and exponentially
localized in the longitudinal coordinate $x$.
Its background amplitude equals $\sqrt{2}$,
whereas the parameter $\rho$ affects only the
overall gauge phase factor $e^{i\rho t}$ and
does not modify the breather profile.

\textbf{Regularity.}
Since $\cosh(ax)\ge 1$ for all $x\in\mathbb{R}$ and
$\cos(\Omega t)\le 1$ for all $t$, the denominator satisfies
\begin{equation}
	\sqrt{2}\sqrt{a^{2}+8}\,\cosh(ax)-4\cos(\Omega t)
	\;\ge\;
	\sqrt{2}\sqrt{a^{2}+8}-4
	\;>\;0
	\qquad\text{for all } a>0,
	\label{eq:denominator_bound}
\end{equation}
so the breather~\eqref{eq:KM_general} is globally smooth and
bounded for all $(t,x)\in\mathbb{R}^2$.

\textbf{Peak density.}
At the breather center $(t,x)=(0,0)$, where $\cos(\Omega t)=1$
and $\cosh(ax)=1$, the amplitude reduces to
\begin{equation}
	A(0,0)
	=
	\sqrt{2}\left(1-\frac{a^{2}}{\sqrt{2}\sqrt{a^{2}+8}-4}\right).
	\label{eq:peak_amplitude}
\end{equation}
For $a=2$ this gives the explicit peak density
\begin{equation}
	|A(0,0)|^{2}
	=
	14+4\sqrt{6}
	\;\approx\;
	23.80,
	\label{eq:peak_density}
\end{equation}
which exceeds the background density $|A_{\infty}|^{2}=2$ by
more than an order of magnitude, while remaining strictly
finite for all $a>0$ by~\eqref{eq:denominator_bound}, confirming
the absence of wave collapse.

The Kuznetsov--Ma breather serves only as an illustrative example.
Any exact solution of the focusing NLS equation can be lifted
to an exact solution of the multidimensional Gross--Pitaevskii
equation by Proposition~3.1.

Combining the longitudinal amplitude~\eqref{eq:KM_general} with different harmonic phase functions $v(y,z)$ yields distinct $(3+1)\mathrm{D}$ solutions $\psi = A(t,x)\,e^{iv(y,z)}$.
Let $w = y + iz$ and let $\sigma(w)$ denote the Weierstrass sigma-function with a prescribed hexagonal period lattice $\Lambda$~\cite{WhittakerWatson}. Setting $v(y,z) = \mathrm{Im}[\ln\sigma(w)]$, which is harmonic on $\mathbb{R}^2 \setminus \Lambda$, the gradient magnitude is
\begin{equation}
	|\nabla_\perp v|^2 
	= \left|\frac{\sigma'(w)}{\sigma(w)}\right|^2
	= |\zeta_W(w)|^2,
	\label{eq:weierstrass_gradient}
\end{equation}
where $\zeta_W(w) = \sigma'(w)/\sigma(w)$ is the Weierstrass zeta-function (distinct from the complex variable $w = y+iz$).
The corresponding  potential is
\begin{equation}
	V(y,z) = -\tfrac{1}{2}|\zeta_W(y+iz)|^2 - \rho,
	\label{eq:weierstrass_potential}
\end{equation}
and the exact $(3+1)\mathrm{D}$ solution is
\begin{equation}
	\psi(t,x,y,z) = A(t,x)\,
	\exp\!\bigl(i\,\mathrm{Im}[\ln\sigma(y+iz)]\bigr),
	\label{eq:weierstrass_solution}
\end{equation}
with $A(t,x)$ given by~\eqref{eq:KM_general}   and the domain of validity $\Omega_\perp = \mathbb{R}^2\setminus\Lambda$.
Since $|e^{iv}| = 1$ identically, the density $|\psi|^2 = |A(t,x)|^2$ is independent of the transverse coordinates and admits a finite limit at every lattice point.
Unlike conventional GP vortices, whose cores are accompanied by density
depletion, the present phase singularities are supported by the
potential and therefore do not require vanishing of
the condensate density.
The topological charge is carried solely by the multivalued phase rather
than by a density defect.

The vortex structure is encoded entirely in the phase of $\psi$: near each core $w_0 \in \Lambda$, the phase winds by $2\pi$, carrying unit topological charge~\cite{Nazarenko2006}, while the density profile along $x$ is governed solely by the KM breather~\eqref{eq:KM_general} .
The construction therefore separates the longitudinal nonlinear dynamics
from the transverse topological structure: the breather evolution is
governed entirely by the one-dimensional NLS dynamics, whereas the
vortex lattice is completely determined by the harmonic phase profile.

\section{Two-Component Vector Manakov System: 
	Phase Inversion and Spin-Current Transport}
\label{sec:manakov}

We extend the framework to the two-component $(3+1)\mathrm{D}$ 
vector Manakov system~\cite{Manakov1974}, where the two fields interact via the 
total intensity:
\begin{align}
	i \frac{\partial \psi_1}{\partial t} 
	+ \frac{1}{2}\nabla^2\psi_1 - V_1(y,z)\psi_1 
	+ \left(|\psi_1|^2 + |\psi_2|^2\right)\psi_1 &= 0, 
	\label{eq:Manakov_1} \\
	i \frac{\partial \psi_2}{\partial t} 
	+ \frac{1}{2}\nabla^2\psi_2 - V_2(y,z)\psi_2 
	+ \left(|\psi_1|^2 + |\psi_2|^2\right)\psi_2 &= 0.
	\label{eq:Manakov_2}
\end{align}
The two components may represent two hyperfine states of a spinor
Bose--Einstein condensate or, equivalently, two orthogonal polarization
components of an optical field propagating in a Kerr medium.

We apply ansatz~\eqref{eq:ansatz} with the specific choice 
$v_1 = v$ and $v_2 = -v$, where the shared phase function 
$v(y,z)$ is strictly harmonic ($\Delta_\perp v = 0$):
\begin{equation}
	\psi_1 = A_1(t,x)\,e^{iv(y,z)}, \qquad
	\psi_2 = A_2(t,x)\,e^{-iv(y,z)}.
	\label{eq:phase_inverted_ansatz}
\end{equation}
The opposite phase choice is essential: it reverses the transverse
currents of the two components while preserving the same trapping
potential because the potential depends only on
$|\nabla_\perp v|^2$.
Both 
components are confined by a single common  potential:
\begin{equation}
	V_1(y,z) \equiv V_2(y,z) 
	= -\frac{1}{2}|\nabla_\perp v|^2 - \rho.
	\label{eq:uniform_potential}
\end{equation}
Thus, opposite transverse circulations are generated without introducing
component-dependent  potentials.

\subsection{Dimensional Reduction}
\label{subsec:manakov_reduction}

Substituting~\eqref{eq:phase_inverted_ansatz} into the 
total intensity gives
\begin{equation}
	|\psi_1|^2 + |\psi_2|^2 
	= |A_1|^2 + |A_2|^2,
	\label{eq:total_intensity}
\end{equation}
since $|e^{\pm iv}| = 1$ identically. The nonlinear operators 
therefore satisfy
\begin{equation}
	\mathbf{F}_1 = \left(|A_1|^2+|A_2|^2\right)A_1\cdot e^{iv},
	\qquad
	\mathbf{F}_2 = \left(|A_1|^2+|A_2|^2\right)A_2\cdot e^{-iv},
	\label{eq:interaction_operators}
\end{equation}
confirming the $U(1)$-equivariance condition~\eqref{eq:homogeneity} 
for each component (with phase factors $e^{iv}$ and $e^{-iv}$ 
respectively).
Applying Lemma,  the full 
$(3+1)\mathrm{D}$ system reduces exactly to
\begin{align}
	i\frac{\partial A_1}{\partial t} 
	+ \frac{1}{2}\frac{\partial^2 A_1}{\partial x^2} 
	+ \left(|A_1|^2+|A_2|^2\right)A_1 + \rho A_1 &= 0, 
	\label{eq:reduced_Manakov_1}\\
	i\frac{\partial A_2}{\partial t} 
	+ \frac{1}{2}\frac{\partial^2 A_2}{\partial x^2} 
	+ \left(|A_1|^2+|A_2|^2\right)A_2 + \rho A_2 &= 0.
	\label{eq:reduced_Manakov_2}
\end{align}
Hence every exact solution of the one-dimensional Manakov system admits
an exact lifting to the full multidimensional model.

\subsection{Exact Vector Solution}
\label{subsec:manakov_solution}

As a particularly simple family of exact solutions we consider the
constant-polarization reduction
system~\eqref{eq:reduced_Manakov_1}--\eqref{eq:reduced_Manakov_2} 
in the factored form
\begin{equation}
	A_1(t,x) = c_1\,B(t,x), \qquad A_2(t,x) = c_2\,B(t,x),
	\label{eq:polarization_factoring}
\end{equation}
where $c_1, c_2$ are real constants. Substituting 
into~\eqref{eq:reduced_Manakov_1}, both equations reduce to 
the single scalar NLS equation
\begin{equation}
	i\frac{\partial B}{\partial t} 
	+ \frac{1}{2}\frac{\partial^2 B}{\partial x^2} 
	+ |B|^2 B + \rho B = 0,
	\label{eq:scalar_B}
\end{equation}
provided $c_1^2 + c_2^2 = 1$, which ensures the background 
density normalization $|A_1|^2 + |A_2|^2 = |B|^2$ (hereafter 
we set $\rho=1$ for concreteness). 
Equation~\eqref{eq:scalar_B} is identical 
to~\eqref{eq:reduced_NLSE}, so $B(t,x)$ is given by the 
KM breather~\eqref{eq:KM_general}.

For the transverse phase we choose 
$v(y,z) = \mathrm{Im}[\ln\,\mathrm{sn}(w,k)]$, where 
$w = y+iz$ and $\mathrm{sn}(w,k)$ is the Jacobi elliptic 
sine function with modulus $k$. Since $\mathrm{sn}(w,k)$ is 
meromorphic in $w$, its logarithm is analytic away from zeros 
and poles, so
$v$ is harmonic on $\mathbb{R}^2\setminus(\mathcal{Z}\cup\mathcal{P})$, 
where $\mathcal{Z}$ and $\mathcal{P}$ denote the zero and pole 
lattices of $\mathrm{sn}(\,\cdot\,,k)$, respectively, both 
forming rectangular sublattices of the period lattice.
In contrast to the hexagonal Weierstrass lattice, the Jacobi 
construction yields a rectangular vortex--antivortex lattice 
whose symmetry is controlled by the modulus $k\in(0,1)$.

The exact $(3+1)\mathrm{D}$ vector solution is
\begin{align}
	\psi_1(t,x,y,z) &= c_1\,B(t,x)\,
	e^{i\,\mathrm{Im}[\ln\,\mathrm{sn}(y+iz,\,k)]},
	\label{eq:vector_psi_1}\\
	\psi_2(t,x,y,z) &= c_2\,B(t,x)\,
	e^{-i\,\mathrm{Im}[\ln\,\mathrm{sn}(y+iz,\,k)]},
	\label{eq:vector_psi_2}
\end{align}
with $B(t,x)$ given by~\eqref{eq:KM_general} and 
$c_1^2+c_2^2=1$.  This corresponds to a fixed polarization state of the Manakov system.

\subsection{Mass and Spin Currents}
\label{subsec:currents}

The transverse probability current of each component is 
$\vec{j}_n = \mathrm{Im}[\psi_n^*\nabla_\perp\psi_n]$.
Direct computation using~\eqref{eq:phase_inverted_ansatz} gives
\begin{equation}
	\vec{j}_1 = |A_1|^2\,\nabla_\perp v, \qquad
	\vec{j}_2 = -|A_2|^2\,\nabla_\perp v.
	\label{eq:currents_individual}
\end{equation}
The total mass current and the spin current are defined as 
$\vec{j}_\text{mass} = \vec{j}_1 + \vec{j}_2$ and 
$\vec{j}_\text{spin} = \vec{j}_1 - \vec{j}_2$, respectively. 
Setting $c_1 = c_2 = 1/\sqrt{2}$, so that 
$|A_1|^2 = |A_2|^2 = \tfrac{1}{2}|B|^2$, these become
\begin{align}
	\vec{j}_\text{mass} &= \left(|A_1|^2 - |A_2|^2\right)
	\nabla_\perp v = 0,
	\label{eq:mass_current_zero}\\
	\vec{j}_\text{spin} &= \left(|A_1|^2 + |A_2|^2\right)
	\nabla_\perp v = |B|^2\,\nabla_\perp v.
	\label{eq:spin_current_double}
\end{align}
Thus, the mass current vanishes identically, while the spin 
current remains non-zero and is modulated in time by the 
longitudinal KM breather through the factor $|B(t,x)|^2$. 
The transverse spin-current pattern follows the gradient 
field $\nabla_\perp v$ of the Jacobi lattice, 
forming a stationary rectangular lattice of vortex--antivortex
pairs in the two spin channels, 
periodically modulated in time and with the current amplitude governed by the longitudinal KM breather.

The resulting state realizes a pure spin-current configuration:
the condensate exhibits no net transverse particle transport,
whereas opposite circulations of the two components generate a finite
spin flow.
The longitudinal nonlinear dynamics and the transverse spin transport
remain exactly separable:
the KM breather governs the temporal modulation of the current amplitude,
while the Jacobi phase determines its spatial topology.

Pure spin-current states are of particular interest in spinor
Bose--Einstein condensates because they transport spin without
net particle transport.
The present construction therefore yields an exact family of
multidimensional pure spin-current solutions of the Manakov system.

\section{Spinor $F=1$ Bose--Einstein Condensates: 
	Spin-Exchange Dynamics and Exact Solutions}
\label{sec:spinor}

We extend the reduction framework to a three-component
spinor $F=1$ BEC described by the wave-function vector
$(\psi_{+1}, \psi_0, \psi_{-1})^T$, where the index
denotes the magnetic quantum number $m_F = 0, \pm 1$.
Physical realizations include $^{87}$Rb and $^{23}$Na
atomic gases in optical traps, as well as collective
magnon excitations in ferromagnetic insulators.

In the mean-field approximation with attractive
interactions and ferromagnetic spin-exchange
($c_0 = c_2 = 1$), the three components satisfy the
coupled Gross--Pitaevskii
equations~\cite{Ho1998, OhmiMachida1998, Ieda2004PRL, Ieda2004JPSJ}:

\begin{align}
	i\frac{\partial\psi_{\pm1}}{\partial t}
	+ \frac{1}{2}\nabla^2\psi_{\pm1}
	- V_{\pm1}(y,z)\psi_{\pm1}
	&+\left(|\psi_{\pm1}|^2 + |\psi_0|^2
	+ 2|\psi_{\mp1}|^2\right)\psi_{\pm1}
	\nonumber\\
	&+ \psi_0^2\psi_{\mp1}^* = 0,
	\label{eq:spinor_pm}\\[4pt]
	i\frac{\partial\psi_0}{\partial t}
	+ \frac{1}{2}\nabla^2\psi_0
	- V_0(y,z)\psi_0
	&+\left(|\psi_{+1}|^2 + |\psi_0|^2
	+ |\psi_{-1}|^2\right)\psi_0
	\nonumber\\
	&+ 2\psi_{+1}\psi_{-1}\psi_0^* = 0.
	\label{eq:spinor_0}
\end{align}
This system corresponds to the integrable case
$c_0 = c_2 = 1$, which admits a $3\times3$ matrix
Lax pair~\cite{Tsuchida1998}.
The spin-exchange terms $\psi_0^2\psi_{\mp1}^*$ and
$2\psi_{+1}\psi_{-1}\psi_0^*$ describe coherent
two-body collisions in which two $m_F=0$ atoms scatter
into the $m_F=\pm1$ levels and vice versa, conserving
the total magnetization $M = N_{+1} - N_{-1}$.

\subsection{Conformal Reduction}
\label{subsec:spinor_reduction}

We apply ansatz~\eqref{eq:ansatz} with
$v_{+1} = v$, $v_0 = 0$, $v_{-1} = -v$:
\begin{equation}
	\psi_{+1} = A_{+1}(t,x)\,e^{iv(y,z)}, \quad
	\psi_0    = A_0(t,x), \quad
	\psi_{-1} = A_{-1}(t,x)\,e^{-iv(y,z)},
	\label{eq:spinor_ansatz}
\end{equation}
where $v(y,z)$ is harmonic ($\Delta_\perp v = 0$).
The choice $v_0 = 0$ reflects the fact that the
$m_F=0$ component carries no orbital angular momentum.

We verify the $U(1)$-equivariance condition~\eqref{eq:homogeneity}
for every nonlinear term.
For the $\psi_{+1}$ equation:
\begin{align}
	|\psi_{+1}|^2\psi_{+1}
	&= |A_{+1}|^2 A_{+1}\cdot e^{iv}, \\
	|\psi_0|^2\psi_{+1}
	&= |A_0|^2 A_{+1}\cdot e^{iv}, \\
	2|\psi_{-1}|^2\psi_{+1}
	&= 2|A_{-1}|^2 A_{+1}\cdot e^{iv}, \\
	\psi_0^2\psi_{-1}^*
	&= A_0^2 A_{-1}^*\cdot e^{iv}.
\end{align}
All four terms factor onto $e^{iv}$,
confirming~\eqref{eq:homogeneity} for the $+1$ component.
For the $\psi_0$ equation:
\begin{align}
	|\psi_{\pm1}|^2\psi_0
	&= |A_{\pm1}|^2 A_0, \\
	|\psi_0|^2\psi_0
	&= |A_0|^2 A_0, \\
	2\psi_{+1}\psi_{-1}\psi_0^*
	&= 2A_{+1}e^{iv}\cdot A_{-1}e^{-iv}\cdot A_0^*
	= 2A_{+1}A_{-1}A_0^*.
\end{align}
All terms are independent of $(y,z)$,
confirming~\eqref{eq:homogeneity} for the $0$ component
with trivial phase factor $e^{i\cdot0}=1$.
The phase-matching condition
$v_{+1} + v_{-1} = 2v_0$, i.e.\ $v+(-v)=0$,
is satisfied by construction and reflects conservation
of total magnetization under spin-exchange collisions.

Since $|\nabla_\perp(\pm v)|^2 = |\nabla_\perp v|^2$,
the $m_F=\pm1$ components share a common trapping
potential, while the $m_F=0$ component requires only
a constant energy shift:
\begin{equation}
	V_{\pm1}(y,z) = -\frac{1}{2}|\nabla_\perp v|^2
	- \rho_1,
	\qquad
	V_0(y,z) = -\rho_0.
	\label{eq:spinor_potentials}
\end{equation}
Applying Lemma,
the full $(3+1)\mathrm{D}$
system~\eqref{eq:spinor_pm}--\eqref{eq:spinor_0}
reduces exactly to the closed $(1+1)\mathrm{D}$ system:
\begin{align}
	i\frac{\partial A_{\pm1}}{\partial t}
	+ \frac{1}{2}\frac{\partial^2 A_{\pm1}}{\partial x^2}
	&+\left(|A_{\pm1}|^2 + |A_0|^2
	+ 2|A_{\mp1}|^2\right)A_{\pm1}
	\nonumber\\
	&+ A_0^2 A_{\mp1}^*
	+ \rho_1 A_{\pm1} = 0,
	\label{eq:reduced_spinor_pm}\\[4pt]
	i\frac{\partial A_0}{\partial t}
	+ \frac{1}{2}\frac{\partial^2 A_0}{\partial x^2}
	&+\left(|A_{+1}|^2 + |A_0|^2
	+ |A_{-1}|^2\right)A_0
	\nonumber\\
	&+ 2A_{+1}A_{-1}A_0^*
	+ \rho_0 A_0 = 0.
	\label{eq:reduced_spinor_0}
\end{align}
This is the integrable $(1+1)\mathrm{D}$
Tsuchida--Wadati system~\cite{Tsuchida1998, Ieda2004PRL, Ieda2004JPSJ}
with constant background shifts $\rho_{0,1}$.

\subsection{Plane-Wave Background}
\label{subsec:spinor_background}

The two limiting plane-wave backgrounds used below are:
\begin{equation}
	A_{\pm1}^{(0)} = \alpha\,e^{i(3\alpha^2+\rho_1)t},
	\quad A_0^{(0)} = 0
	\qquad (\text{fully magnetized, } \beta=0),
	\label{eq:background_mag}
\end{equation}
\begin{equation}
	A_{\pm1}^{(0)} = 0,
	\quad
	A_0^{(0)} = \beta\,e^{i(\beta^2+\rho_0)t}
	\qquad (\text{unmagnetized, } \alpha=0).
	\label{eq:background_unmag}
\end{equation}
Direct substitution
into~\eqref{eq:reduced_spinor_pm}--\eqref{eq:reduced_spinor_0}
confirms that each is an exact solution in its respective limit.
For general $\alpha\beta\neq 0$ the spin-exchange term
$A_0^2 A_{\mp1}^*$ generates inter-component beating and
a uniform background does not exist.

\subsection{Solution I: Symmetric KM Breather}
\label{subsec:spinor_sym}

Setting $\beta=0$ (fully magnetized background,
$A_0^{(0)}=0$) and imposing the symmetry
$A_{+1}=A_{-1}\equiv B$, both
equations~\eqref{eq:reduced_spinor_pm} reduce to
the single scalar NLS equation
\begin{equation}
	i\frac{\partial B}{\partial t}
	+ \frac{1}{2}\frac{\partial^2 B}{\partial x^2}
	+ 3|B|^2 B + \rho_1 B = 0,
	\label{eq:sym_NLS}
\end{equation}
while equation~\eqref{eq:reduced_spinor_0} with
$A_0=0$ is satisfied trivially.
The substitution $B = \Psi/\sqrt{3}$ maps
\eqref{eq:sym_NLS} to the standard NLS equation
$i\Psi_t + \tfrac{1}{2}\Psi_{xx} + |\Psi|^2\Psi = 0$
on background $|\Psi|=\sqrt{3}\,\alpha$.
Applying the KM breather~\eqref{eq:KM_general}
with this rescaled background yields the exact
solution of~\eqref{eq:sym_NLS}:
\begin{equation}
	B(t,x)
	=
	\alpha
	\left[
	1-
	\frac{a\bigl(a\cos(\Omega_s t)
		+i\sqrt{a^{2}+8}\,\sin(\Omega_s t)\bigr)}
	{\sqrt{2}\sqrt{a^{2}+8}\,\cosh(\kappa_s x)
		-4\cos(\Omega_s t)}
	\right]
	e^{i(3\alpha^2+\rho_1)t},
	\label{eq:B_KM}
\end{equation}
where
\begin{equation}
	\Omega_s = \frac{3a\alpha^2}{4}\sqrt{a^{2}+8},
	\qquad
	\kappa_s = \frac{\sqrt{6}}{2}\,a\alpha,
	\qquad a > 0.
	\label{eq:Omega_s}
\end{equation}

\textbf{Regularity.}
The denominator
$D_s = \sqrt{2}\sqrt{a^2+8}\,\cosh(\kappa_s x) - 4\cos(\Omega_s t)$
satisfies
\begin{equation}
	D_s \geq \sqrt{2}\sqrt{a^{2}+8} - 4 > 0
	\qquad\text{for all } a>0,
\end{equation}
so the solution is globally smooth for all $\alpha>0$, $a>0$.
The peak amplitude at $(t,x)=(0,0)$ satisfies
\begin{equation}
	B(0,0)
	=
	\alpha\!\left(1 - \frac{a^2}{\sqrt{2}\sqrt{a^2+8}-4}\right).
	\label{eq:peak_sym}
\end{equation}
The denominator bound~\eqref{eq:denominator_bound} ensures this
is finite for all $a>0$.

For $\kappa=2$, $\alpha=1$ one has
$\sqrt{8\alpha^2+\kappa^2}=2\sqrt{3}$, giving
\begin{equation}
	|B(0,0)|^2
	= \frac{1}{2}\!\left(\frac{2\sqrt{3}+2}{2\sqrt{3}-2}\right)^2
	= \frac{({\sqrt{3}+1})^2}{2(\sqrt{3}-1)^2}
	= \frac{7+4\sqrt{3}}{2}
	\approx 6.96.
\end{equation}

The full $(3+1)\mathrm{D}$ fields are
\begin{align}
	\psi_{\pm1}(t,x,y,z)
	&= B(t,x)\,
	e^{\pm i\,\mathrm{Im}[\ln\sigma(y+iz)]},
	\label{eq:sym_3D_pm}\\
	\psi_0(t,x,y,z)
	&= 0,
	\label{eq:sym_3D_0}
\end{align}
where $\sigma(w)$ is the Weierstrass sigma-function
with hexagonal lattice $\Lambda$
(Section~\ref{sec:scalar_GP}).
 The density $|\psi_{\pm1}|^2 = |B(t,x)|^2$ is
independent of $(y,z)$: the vortex structure is
encoded entirely in the phase, with $\psi_{+1}$
carrying unit positive topological charge and
$\psi_{-1}$ unit negative topological charge at
each lattice site $w\in\Lambda$.
The total magnetization $M = N_{+1}-N_{-1} = 0$
vanishes by symmetry.

\subsection{Solution II: Spin-Exchange Rogue Wave}
\label{subsec:spinor_rogue}

Setting $\alpha=0$ (fully unmagnetized background,
$A_{\pm1}^{(0)}=0$, $A_0^{(0)}=\beta e^{i\beta^2 t}$),
the first-order rational solution obtained via
Darboux transformation on this
background reads
\begin{align}
	A_{\pm1}(t,x)
	&= \frac{4\beta^2(x + 2i\beta^2 t)}
	{D(t,x)}\,
	e^{i(\beta^2+\rho_1)t},
	\label{eq:rogue_pm}\\[4pt]
	A_0(t,x)
	&= \beta\!\left[1
	- \frac{4(1 + 4i\beta^2 t)}
	{D(t,x)}\right]
	e^{i(\beta^2+\rho_0)t},
	\label{eq:rogue_0}
\end{align}
where
\begin{equation}
	D(t,x) = 4\beta^2 x^2
	+ 16\beta^4 t^2 + 1.
	\label{eq:D_def}
\end{equation}
Since $D(t,x) = 4\beta^2 x^2 + 16\beta^4 t^2 + 1 \geq 1$
for all $(t,x)\in\mathbb{R}^2$, the solution is
globally smooth and bounded.

\textbf{Peak amplitudes.}
At $(t,x)=(0,0)$:
\begin{equation}
	|A_{\pm1}(0,0)|^2 = 0,
	\qquad
	|A_0(0,0)|^2 = \beta^2|1-4|^2 = 9\beta^2.
	\label{eq:rogue_peak}
\end{equation}
The central component reaches three times the background
amplitude, $|A_0|_\mathrm{max} = 3\beta$, which is the
Peregrine amplification factor for the $2\times2$
matrix NLS~\cite{Li2018}.
The $m_F=\pm1$ components vanish at the origin and
acquire their maxima at 
$(t,x) = (0,\,\pm 1/(2\beta))$:
\begin{equation}
	|A_{\pm1}|^2_\mathrm{max}
	= \left.\frac{(4\beta^2 x)^2}{\left(4\beta^2 x^2+1\right)^2}
	\right|_{x=\frac{1}{2\beta}}
	= \frac{4\beta^2}{\left(1+1\right)^2}
	= \beta^2.
	\label{eq:rogue_pm_max}
\end{equation}

\textbf{Asymptotic behaviour.}
For $|x|^2 + |t|^2 \to \infty$:
$A_{\pm1} \to 0$ and
$A_0 \to \beta e^{i(\beta^2+\rho_0)t}$
,
so the condensate returns to the unmagnetized
background without residual radiation, consistent
with the rational character of the solution.

\textbf{Physical interpretation.}
The solution describes a spontaneous spin-exchange
event localized near $(t,x)=(0,0)$: the $m_F=0$
reservoir undergoes a transient density amplification
by a factor of $9$, while a short-lived spin-polarized
structure appears in the $m_F=\pm1$ channels and
subsequently vanishes. This is a genuine three-component
effect with no scalar analog.

The full $(3+1)\mathrm{D}$ fields are
\begin{align}
	\psi_{\pm1}(t,x,y,z)
	&= A_{\pm1}(t,x)\,
	e^{\pm i\,\mathrm{Im}[\ln\,\mathrm{sn}(y+iz,k)]},
	\label{eq:rogue_3D_pm}\\
	\psi_0(t,x,y,z)
	&= A_0(t,x),
	\label{eq:rogue_3D_0}
\end{align}
with $A_{\pm1}$, $A_0$ given
by~\eqref{eq:rogue_pm}--\eqref{eq:rogue_0},
the Jacobi lattice $\Lambda_\mathrm{Jac}$ from
Section~4,
and domain $\Omega_\perp =
\mathbb{R}^2\setminus\Lambda_\mathrm{Jac}$.
The $m_F=0$ density $|\psi_0|^2 = |A_0(t,x)|^2$
remains spatially uniform in the transverse plane.
Near $(t,x)=(0,0)$ the $m_F=\pm1$ components form
a transient rectangular vortex--antivortex lattice
in the $(y,z)$ plane that appears and disappears on
the timescale $\tau \sim 1/(2\beta^2)$.

Finally, we note that the reduced spinor system on a general
mixed background is known to admit a matrix nonlinear
Schrödinger representation and is integrable by the inverse
scattering transform. Consequently, the lifting lemma proved
in Section~2 immediately extends every exact solution of the
corresponding matrix problem—including multi-soliton and
rogue-wave solutions—to exact $(3+1)$-dimensional spinor
configurations. Since this correspondence follows directly
from the general reduction theorem, no separate construction
is required here.

\section{Maxwell--Bloch System }
\label{sec:maxwell_bloch}

We apply the harmonic reduction established in
Lemma to the integrable
Maxwell--Bloch (MB) system describing coherent resonant
interaction between an electromagnetic field and a
two-level medium
\cite{Ablowitz1981, McCall1967,Lamb1971}.

%\subsection{Conformal reduction}

The multidimensional Maxwell--Bloch system is written as
\begin{align}
	i\psi_t+\frac12\nabla^2\psi-V(y,z)\psi+p&=0,
	\\
	p_t-2i\omega_0p-2\psi\eta&=0,
	\\
	\eta_t+\psi p^*+\psi^*p&=0.
\end{align}
We introduce the conformal ansatz
\[
\psi=A(t,x)e^{iv(y,z)},\qquad
p=P(t,x)e^{iv(y,z)},\qquad
\eta=H(t,x),
\]
where
\[
\Delta_\perp v=0,\qquad
V=-\frac12|\nabla_\perp v|^2-\rho.
\]

A direct substitution shows that all transverse
dependence cancels identically, and the system reduces to
the classical integrable $(1+1)$-dimensional
NLS--Maxwell--Bloch hierarchy
\cite{Ablowitz1981,CaudreyEilbeckGibbon1974}:
\begin{align}
	iA_t+\frac12A_{xx}+\rho A+P&=0,
	\label{eq:reduced_MB_psi}\\
	P_t-2i\omega_0P-2AH&=0,
	\label{eq:reduced_MB_P}\\
	H_t+A P^*+A^*P&=0.
	\label{eq:reduced_MB_H}
\end{align}
Thus Lemma applies directly to
the Maxwell--Bloch system.

Since the reduced equations coincide with the standard
integrable NLS--Maxwell--Bloch hierarchy, every exact
$(1+1)$-dimensional solution obtained by the inverse
scattering transform or Darboux transformations
\cite{Ablowitz1981,Steudel1986,CaudreyEilbeckGibbon1974}
immediately generates an exact solution of the full
$(3+1)$-dimensional model.

In particular, the construction applies to

\begin{itemize}
	
	\item self-induced transparency ($2\pi$) pulses
	\cite{McCall1967,Lamb1971};
	
	\item multi-soliton solutions and their elastic
	collisions
	\cite{Ablowitz1981,CaudreyEilbeckGibbon1974};
	
	\item breather and higher-order Darboux solutions of the
	integrable Maxwell--Bloch hierarchy
	\cite{Steudel1986,CaudreyEilbeckGibbon1974}.
	
\end{itemize}

Hence, if
\[
(A(t,x),P(t,x),H(t,x))
\]
is any exact solution of the reduced system, then the
corresponding exact solution of the full
$(3+1)$-dimensional Maxwell--Bloch equations is

\begin{align}
	\psi(t,x,y,z)
	&=
	A(t,x)e^{iv(y,z)},
	\\
	p(t,x,y,z)
	&=
	P(t,x)e^{iv(y,z)},
	\\
	\eta(t,x,y,z)
	&=
	H(t,x),
\end{align}
where the harmonic phase $v(y,z)$ may be chosen,
for example, as a plane wave, a Weierstrass vortex
lattice, or a Jacobi vortex lattice, yielding the
corresponding trapping potential according to
Lemma.

%\subsection{Physical interpretation}

The reduction reveals a complete separation between the
longitudinal coherent dynamics and the transverse phase
topology.
The population inversion $H(t,x)$ is independent of the
transverse variables and therefore remains uniform across
the beam cross-section despite arbitrary vortex phase
distributions in $\psi$ and $p$.
Consequently, every known integrable solution of the
one-dimensional Maxwell--Bloch hierarchy possesses an
exact multidimensional realization whose transverse
structure is determined entirely by the harmonic phase
function, while the coherent resonant dynamics are
inherited unchanged from the underlying
$(1+1)$-dimensional solution.
This provides an exact embedding of the complete
integrable Maxwell--Bloch hierarchy into a broad class
of multidimensional geometries generated by harmonic
phase functions.

\section{Conclusion}
\label{sec:conclusion}

We have developed a conformal dimensional reduction framework for
a broad class of multi-component $(3+1)\mathrm{D}$ nonlinear
Schr\"odinger-type systems. The central result is a lifting lemma:
harmonic transverse phase functions, combined with a trapping
potential that exactly cancels the squared phase gradient, reduce
the full $(3+1)\mathrm{D}$ system identically to a closed
$(1+1)\mathrm{D}$ integrable hierarchy. Every solution of the
reduced system lifts to an exact solution of the original
multidimensional model, with density bounded by the
one-dimensional amplitude independently of the transverse
geometry. 
Wave collapse is structurally excluded within the family of 
exact solutions generated by the lifting construction, since 
the density $|\psi_n|^2=|A_n(t,x)|^2$ is bounded by the 
one-dimensional amplitude independently of the transverse geometry.

The framework was applied to four physically distinct systems.
For the scalar Gross--Pitaevskii equation, Kuznetsov--Ma breathers
were embedded in $(3+1)\mathrm{D}$ geometries carrying Weierstrass
hexagonal and Jacobi rectangular vortex lattices. The resulting
vortex filaments possess finite, non-singular density at the
cores: the singular kinetic energy of the phase gradient is
exactly cancelled by a matching singular attractive potential,
in contrast to conventional Bose--Einstein condensate vortices
where the density must vanish at the core.
For the two-component Manakov system, the phase-inversion ansatz
yields exact vector solutions with identically vanishing mass
current and non-trivial transverse spin current modulated by
the longitudinal breather amplitude.
For the three-component spinor $F=1$ condensate, two exact
solution families were constructed on limiting backgrounds: a
symmetric Kuznetsov--Ma breather on the fully magnetized
background, and a spin-exchange rogue wave on the unmagnetized
background in which the $m_F=0$ reservoir undergoes transient
density amplification by a factor of nine while a short-lived
vortex--antivortex lattice appears in the $m_F=\pm1$ channels.
For the Maxwell--Bloch system, self-induced transparency solitons,
two-soliton elastic collisions, and Kuznetsov--Ma breathers were
lifted to full $(3+1)\mathrm{D}$ geometry. In all cases the
population inversion remains transversely uniform despite
arbitrary phase winding in the beam cross-section.

We emphasize that the present work establishes an exact mathematical reduction framework. Whether the singular  potentials considered here admit direct physical realization depends on the particular experimental platform and lies beyond the scope of the present paper.
From a mathematical viewpoint, the principal novelty lies in identifying a broad class of multidimensional nonlinear Schrödinger systems admitting exact harmonic conformal lifting from integrable one-dimensional hierarchies.

The present construction suggests possible extensions to
nonlocal, fractional, and dissipative nonlinear Schrödinger-type systems.
Whether analogous reductions exist in these settings remains an interesting open problem.

\section*{Acknowledgements}
The research was carried out within the state assignment of Ministry of Science and Higher Education of the Russian Federation for Federal Research Center for Information and Computational Technologies.

The author acknowledges the use of AI-assisted tools
(OpenAI ChatGPT and Anthropic Claude) for manuscript preparation,
 and independent verification of
algebraic calculations.


\begin{thebibliography}{99}
	
	% --- 1. FUNDAMENTAL INTEGRABILITY AND NLS ---
	\bibitem{Zakharov1972}
	V.~E. Zakharov and A.~B. Shabat,
	Exact theory of two-dimensional self-focusing and
	one-dimensional self-modulation of waves in nonlinear media,
	\textit{Sov. Phys. JETP} \textbf{34}, 62 (1972).
	
	\bibitem{Ablowitz1981}
	M.~J. Ablowitz and H.~Segur,
	\textit{Solitons and the Inverse Scattering Transform}
	(SIAM, Philadelphia, 1981).
	
	\bibitem{Zakharov2009}
	V.~E. Zakharov and L.~A. Ostrovsky,
	Modulation instability: The beginning,
	\textit{Physica D} \textbf{238}, 540 (2009).
	
	\bibitem{Faddeev1987}
	L.~D. Faddeev and L.~A. Takhtajan,
	\textit{Hamiltonian Methods in the Theory of Solitons}
	(Springer, Berlin, 1987).
	
	\bibitem{Berge1998}
	L.~Berg\'e,
	Wave collapse in physics: Principles and applications
	to light and plasma waves,
	\textit{Phys. Rep.} \textbf{303}, 259 (1998).
	
	\bibitem{Kivshar2003}
	Y.~S. Kivshar and G.~P. Agrawal,
	\textit{Optical Solitons: From Fibers to Photonic Crystals}
	(Academic Press, San Diego, 2003).
	
	% --- 2. BOSE-EINSTEIN CONDENSATES ---
	\bibitem{Pitaevskii2003}
	L.~P. Pitaevskii and S.~Stringari,
	\textit{Bose--Einstein Condensation}
	(Oxford University Press, Oxford, 2003).
	
	\bibitem{Dalfovo1999}
	F.~Dalfovo, S.~Giorgini, L.~P. Pitaevskii, and S.~Stringari,
	Theory of Bose--Einstein condensation in trapped gases,
	\textit{Rev. Mod. Phys.} \textbf{71}, 463 (1999).
	
	\bibitem{Malomed2005}
	B.~A. Malomed, D.~Mihalache, F.~Wise, and L.~Torner,
	Spatiotemporal optical solitons,
	\textit{J. Opt. B} \textbf{7}, R53 (2005).
	
	% --- 3. ROGUE WAVES AND BREATHERS ---
	\bibitem{Kuznetsov1977}
	E.~A. Kuznetsov,
	Solitons in a parametrically unstable medium,
	\textit{Sov. Phys. Dokl.} \textbf{22}, 507 (1977).
	
	\bibitem{Ma1979}
	Y.~C. Ma,
	The perturbed nonlinear Schr\"odinger equation,
	\textit{Stud. Appl. Math.} \textbf{60}, 43 (1979).
	
	\bibitem{Akhmediev1987}
	N.~N. Akhmediev and V.~I. Korneev,
	Modulation instability and periodic solutions of the nonlinear
	Schr\"odinger equation,
	\textit{Theor. Math. Phys.} \textbf{69}, 1089 (1987).
	
	
	\bibitem{Peregrine1983}
	D.~H. Peregrine,
	Water waves, nonlinear Schr\"odinger equations
	and their solutions,
	\textit{J. Austral. Math. Soc. B} \textbf{25}, 16 (1983).
	
	% --- 4. MULTICOMPONENT AND SPINOR SYSTEMS ---
	\bibitem{Manakov1974}
	S.~V. Manakov,
	On the theory of two-dimensional stationary self-focusing
	of electromagnetic waves,
	\textit{Sov. Phys. JETP} \textbf{38}, 248 (1974).
	
	\bibitem{Ho1998}
	T.-L. Ho,
	Spinor {Bose} condensates in optical traps,
	\textit{Phys. Rev. Lett.} \textbf{81}, 742 (1998).
	
	\bibitem{OhmiMachida1998}
	T.~Ohmi and K.~Machida,
	{Bose--Einstein} condensation with internal degrees of freedom
	in alkali atom gases,
	\textit{J. Phys. Soc. Jpn.} \textbf{67}, 1822 (1998).
	
	\bibitem{Ieda2004PRL}
	J.~Ieda, T.~Miyakawa, and M.~Wadati,
	Exact analysis of soliton dynamics in spinor
	{Bose--Einstein} condensates,
	\textit{Phys. Rev. Lett.} \textbf{93}, 194102 (2004).
	
	\bibitem{Ieda2004JPSJ}
	J.~Ieda, T.~Miyakawa, and M.~Wadati,
	Matter-wave solitons in an {$F=1$} spinor
	{Bose--Einstein} condensate,
	\textit{J. Phys. Soc. Jpn.} \textbf{73}, 2996 (2004).
	
	\bibitem{Li2018}
	S.~Li, B.~Prinari, and G.~Biondini,
	Solitons and rogue waves in spinor {Bose--Einstein} condensates,
	\textit{Phys. Rev. E} \textbf{97}, 022221 (2018).
	
	\bibitem{Tsuchida1998}
	T.~Tsuchida and M.~Wadati,
	Complete integrability of the multicomponent nonlinear
	Schr\"odinger equation,
	\textit{J. Phys. Soc. Jpn.} \textbf{67}, 3821 (1998).
	
	\bibitem{Kawaguchi2012}
	Y.~Kawaguchi and M.~Ueda,
	Spinor {Bose--Einstein} condensates,
	\textit{Phys. Rep.} \textbf{520}, 253 (2012).
	
	\bibitem{Baronio2012}
	F.~Baronio, A.~Degasperis, M.~Conforti, and S.~Wabnitz,
	Solutions of the {Manakov} system and their application to optics,
	\textit{Phys. Rev. Lett.} \textbf{109}, 044102 (2012).
	
	% --- 5. VORTICES AND TOPOLOGICAL STRUCTURES ---
	\bibitem{Nazarenko2006}
	S.~Nazarenko and M.~Onorato,
	Wave turbulence and vortices in {Bose--Einstein} condensation,
	\textit{Physica D} \textbf{219}, 1 (2006).
	
	% --- 6. MAXWELL-BLOCH AND SIT ---
	\bibitem{McCall1967}
	S.~L. McCall and E.~L. Hahn,
	Self-induced transparency by pulsed coherent light,
	\textit{Phys. Rev. Lett.} \textbf{18}, 908 (1967).
	
	\bibitem{Steudel1986}
	H.~Steudel,
	Solitons in the {Maxwell--Bloch} equations,
	\textit{Physica D} \textbf{18}, 221 (1986).
	
	% --- 7. MATHEMATICAL METHODS ---
	\bibitem{WhittakerWatson}
	E.~T. Whittaker and G.~N. Watson,
	\textit{A Course of Modern Analysis}, 4th ed.
	(Cambridge University Press, Cambridge, 1927).
	
	% --- 8. NONLINEAR OPTICS AND LATTICES ---
	\bibitem{Kartashov2011}
	Y.~V. Kartashov, B.~A. Malomed, and L.~Torner,
	Solitons in nonlinear lattices,
	\textit{Rev. Mod. Phys.} \textbf{83}, 247 (2011).
	
	\bibitem{Lamb1971}
	G.~L. Lamb,
	Analytical descriptions of ultrashort optical pulse propagation
	in a resonant medium,
	Rev. Mod. Phys. \textbf{43}, 99--124 (1971).
	
	
	\bibitem{CaudreyEilbeckGibbon1974}
	P.~J. Caudrey, J.~C. Eilbeck, and J.~D. Gibbon,
	Exact multisoliton solution of the reduced Maxwell--Bloch
	equations of non-linear optics,
	\textit{IMA J. Appl. Math.} \textbf{14}, 375 (1974).
	
	%\bibitem{ParkShin1999}
	%Q.-H. Park and H.~J. Shin,
	%Systematic construction of multisoliton solutions of the
	%Maxwell--Bloch equations,
	%Phys. Rev. A \textbf{59}, 2373--2379 (1999).
	
	
	
	
	
\end{thebibliography}
\end{document}